\documentclass[11pt]{article}
\usepackage{amsthm,amsmath,epsf,amssymb,enumerate,array,color}
\usepackage[all,matrix,arrow,import,arc,knot,poly,color,ps,dvips]{xy}	\CompileMatrices

\definecolor{dgreen}{RGB}{34, 139, 34} \definecolor{webdkgreen}{rgb}{0,0.3,0} \definecolor{blueviolet}{RGB}{138,43,226}
\definecolor{brown}{rgb}{.6,0,0} \definecolor{dblue}{rgb}{0,0,.7} \definecolor{indigo}{RGB}{50,0,105} 

\usepackage[dvips, hyperindex=true, colorlinks=true, linkcolor=dblue,   citecolor=indigo, pagecolor=dblue,
bookmarksopen=false, urlcolor=dblue, final, linktocpage]{hyperref}

\newtheorem{thm}{Theorem}[section]        \newtheorem{lemma}[thm]{Lemma}	
\newtheorem{definition}[thm]{Definition} \newtheorem{prop}[thm]{Proposition}   
\newtheorem{conj}[thm]{Conjecture}  		
\theoremstyle{plain}

\DeclareFontFamily{U}{rsf}{} \DeclareFontShape{U}{rsf}{m}{n}{  <5> <6> rsfs5 <7> <8> <9> rsfs7 <10-> rsfs10}{}
\DeclareMathAlphabet\Scr{U}{rsf}{m}{n} \DeclareMathAlphabet\mathbi{U}{cmr}{bx}{it}

\newcommand{\R}{{\mathrm R}}	\newcommand{\LL}{{\mathrm L}}
\def\CY{Calabi-Yau}
\def\roof{\mbox{\tiny \mbox{$\!\vee$}}}	\def\comp{\mbox{\scriptsize \mbox{$\,\circ \,$}}}
\def\O{\mathcal{O}} \def\c#1{\mathcal{#1}}	
\def\C{{\mathbb C}}\def\P{{\mathbb P}} \def\F{{\mathbb F}}
\def\Z{{\mathbb Z}}
\def\D{\mathbf{D}}
\def\iso{\cong} 
\def\H{\operatorname{H}}
    
\def\id{\operatorname{id}}
\def\Hom{\operatorname{Hom}} \def\sHom{\operatorname{\Scr{H}\!\!\textit{om}}}	
\def\Ext{\operatorname{Ext}}      		
	\def\Aut{\operatorname{Aut}}

\def\ch{\operatorname{\mathrm{ch}}}

\def\Ltensor{\mathbin{\overset{\mbox{\tiny \mbox{$\mathbf L$}} }\otimes}}

\def\ms#1{\mathsf{#1}}		\def\cal{\mathcal}

\def\Cone#1{\operatorname{Cone}\left( #1 \right)}
\def\ses#1#2#3{\xymatrix@1{0 \ar[r] & #1 \ar[r] & #2 \ar[r] & #3 \ar[r] & 0}}
\newcommand{\T}[1]{{\mathsf T}_{#1}} \def\L#1{{\mathsf L}_{#1}}

\newcommand{\mono}{\hookrightarrow} 
\newcommand{\s}[1]{\mathcal{#1}} 

\begin{document}

\title{\bf Quantum symmetries and exceptional collections	\\[10mm]}
\author{{\bf Robert L.~Karp}\thanks{rlk at vt.edu	}	\\[2mm]
\normalsize  Department of Physics,  Virginia Tech\\
\normalsize Blacksburg, VA 24061 USA				}
\date{}
\maketitle

\vskip 1cm

\begin{abstract}
We study the interplay between  discrete quantum symmetries at certain points in the moduli space of \CY\ compactifications, and the associated identities that the geometric realization of D-brane monodromies must satisfy. We show that in a wide class of examples, both local and compact, the monodromy identities in question always follow from a single mathematical statement. One of the simplest examples is the $\Z_5$ symmetry at the Gepner point of the quintic, and the associated D-brane monodromy identity.
\end{abstract}

\vfil\break
\tableofcontents

\section{ Introduction}    \label{s:intro}

Studying B-type topological D-Branes using the derived category allows us to go beyond the picture of D-branes as vector bundles over submanifolds, and opens the window toward understanding various $\alpha'$-corrections. At the same time this technology is very efficient at studying certain problems, like the superpotential \cite{Aspinwall:2004bs,en:Ema}, which seem hard by traditional boundary conformal field theory (CFT) techniques.

{From} the point of view of {\em strings} in string theory, the
appearance of the derived category is  intriguing, but D-branes mandate the categorical approach \cite{Douglas:2000gi, Sharpe:1999qz}. In
particular, B-type topological D-branes are
objects in the bounded derived category of coherent sheaves.  The A-type D-branes have a very different
description, involving the derived Fukaya category. Mirror symmetry
exchanges the A and B branes, and naturally leads to Kontsevich's {\em
homological mirror symmetry} (HMS) conjecture.  For a detailed
exposition of these ideas we refer the reader to the recent book
\cite{DBook}, or the review articles \cite{Paul:TASI2003,Sharpe:2006vd}.

The fact that B-type D-branes undergo monodromy as one moves in the moduli space of complexified Kahler forms is  expressed quite naturally in this language. This is in fact a surprisingly rich area, where the interplay between abstract mathematics (autoequivalences of derived categories) and string theory (discrete symmetries in CFT's) is particularly evident. The main motivation of the present paper is to further our understanding in this area.

To motivate our result we need to start with mirror symmetry in its pre-HMS
phase. In this form mirror symmetry is an isomorphism between the
(complexified) Kahler moduli space $\c M_K(X)$ of a \CY\ variety $X$ and the
moduli space of complex deformations $\c M_c(\widetilde{X})$ of its mirror
$\widetilde{X}$.
For the precise definitions we refer to the book by Cox and
Katz \cite{Cox:Katz}.

The complexified Kahler moduli
space $\c M_K(X)$ is an intricate object, but for $X$ a
hypersurface in a toric variety it has a rich combinatorial structure
and is relatively well-understood.  In particular, the fundamental
group of $\c M_K(X)$ in general is non-trivial, and one can talk about
various monodromy representations. More concretely, there are two
types of boundary divisors in $\c M_K(X)$: ``large radius divisor''
and the ``discriminant''. Both  are reducible
in general.

At a large radius divisor certain cycles of $X$ (or $X$ itself), viewed
as a Kahler manifold, acquire infinite volume. The discriminant is
somewhat harder to describe.  The original definition is that the CFT
associated to a string probing $X$ becomes singular at such a point in
moduli space.  Generically this happens because some D-brane (or
several of them, even infinitely many) becomes massless, and therefore
the effective CFT description provided by the strings fails. A
consequence of this fact is that, by using the mirror map isomorphism
of the moduli spaces, as one approaches the discriminant in $\c
M_K(X)$ one is moving in $\c M_c(\widetilde{X})$ to a point where the
mirror $\widetilde{X}$ is developing a singularity.

Armed with this picture of $\c M_K(X)$, we can fix a basepoint $O$,
and look at loops in $\c M_K(X)$ based at $O$.  Traversing such  loops the D-branes will undergo monodromies, similarly to the BPS particles in Seiberg-Witten theory, and for topological B-branes this leads to
non-trivial functors $\D(X)\to \D(X)$,\footnote{$\D(X)$ will always denote the {\em bounded} derived category of the the variety (or smooth stack) $X$.} which are in fact
equivalences.  Therefore we arrive at a group
homomorphism, the monodromy representation, first suggested by Kontsevich:\footnote{Kontsevich's
ideas were generalized by Horja and Morrison. We refer to \cite{en:Horja} for more details on the history of this topic.}
\begin{equation*}
 \mu \colon \pi_1(\c M_K(X)) \longrightarrow \Aut(\D(X)).
\end{equation*}
At present writing very little is known about $\mu$. The question at hand is: given a pointed loop
in $\c M_K(X)$, what is the associated autoequivalence in $ \D(X)$?
Progress in this direction was made in \cite{en:Horja}, 
where this question is answered for the EZ-degenerations 
introduced in \cite{Horj:EZ}. 

It is clear now that given a presentation of $\pi_1(\c M_K(X))$ where we know the images under $\mu$ of the generators,
the relations in the presentation will determine interesting identities 
in $\Aut(\D(X))$. In particular, whenever one is at a point in moduli space which is an orbifold of some sort (like a Landau-Ginzburg orbifold), the moduli space locally is an orbifold itself. The fact that moduli spaces are in general stacks rather than varieties, precisely because of the appearance of additional automorphisms at different points, complicates matters a bit, as we will see in the example of the next paragraph. But it is clear that there are loops encircling the orbifold point in moduli space which are finite order. Therefore the associated monodromy operator in $\Aut(\D(X))$ has to satisfy an analogous relation. Understanding these relations in $\Aut(\D(X))$ is the goal of this paper. We will find that in a broad range of examples of ``toric'' \CY\ varieties, both local and compact, finite orderness always follows from a general statements concerning Seidel-Thomas twist functors and complete exceptional collections (Prop.~\ref{pullback} and Prop.~\ref{pushout} are two special cases).

For illustration, let us look at the example of the quintic 3-fold in $\P^4$. In this case the compactification $\overline{\c M}_K(X)$ of the Kahler moduli space $\c M_K(X)$ is isomorphic to $\P^1$. $\c M_K(X)$ is also isomorphic to $\c M_c(\widetilde{X})$, the complex structure moduli space of the mirror. In either of these moduli spaces we have three distinguished points:
\begin{enumerate}
 \item $P_{LV}$ is  the large volume limit point in $\c M_K(X)$. It also corresponds to the large complex structure limit point (with maximally unipotent monodromy) in $\c M_c(\widetilde{X})$.
\item $P_0$ is the conifold point. Here the D6-brane wrapping $X$ becomes massless, and therefore the effective CFT description brakes down. Alternatively, the mirror   family $\widetilde{X}$ develops rational double points, in physics language conifolds, and is singular.
\item $P_{LG}$ is the Gepner point, and is a Landau-Ginzburg (LG) orbifold. At this point in moduli space the mirror $\widetilde{X}$ has the Fermat form, and has an additional $\Z_5$ automorphisms. In both formulations  we see that at this point the moduli space has a stacky $\Z_5$ structure.\footnote{Mathematically the compactified moduli space $\overline{\c M}_K(X) \cong \overline{\c M}_c (\widetilde{X})\cong\P^1$ is only a coarse moduli space, while the moduli stack is $\P^1(5,1)$.
As a scheme $\P^1(5,1)\cong \P^1$, but not as a stack.}
\end{enumerate}

Let $M_P$ denote the monodromy associated to a loop around the point
$P$. Since $P_{LV}$ and $P_0$ are the only limit points of $\c
M_K(X)$, and the compactification of this is isomorphic to $\P^1$ (see
\cite{Cox:Katz}), with $\pi_1(\P^1- \{\mbox{2 points} \})=\Z$, one
would want to conclude, incorrectly, that $M_{P_{LV}}$ and $M_{P_0}$
are related.  But as we discussed,
$P_{LG}$ is a stacky point in the moduli space, with finite stabilizer,
and so, at best, the $5$-th power of $M_{P_{LG}}\iso M_{P_{LV}}\comp
M_{P_0}$ is the identity.  This was proposed by
Kontsevich, who checked it in K-theory.  Later  Aspinwall \cite{Paul:TASI2003} realized that in fact
\begin{equation}\label{as1}
 M_{P_{LG}}^5\iso (-)[2].
\end{equation}

The authors of \cite{en:Alberto} observed that the exceptional collection
\begin{equation*}
 \O_{\P^4},\O_{\P^4}(1),\ldots , \O_{\P^4}(4)
\end{equation*}
and its dual collection $\{\Omega_{\P^4}^k(k)\}_{k=0}^4$ ($\Omega_{\P^4}^k$ is the $k$th wedge power of the holomorphic cotangent bundle) were implicit in Aspinwall's proof, and this was the aspect of the proof that gave a handle for generalizations. This observation allowed \cite{en:Alberto} to show that (\ref{as1}) generalizes to \CY\ hypersurfaces in weighted projective spaces of arbitrary dimensions. More precisely, for the generic \CY\ hypersurface $X$ in weighted projective space $\P^n_{w_0\ldots w_n}$, one has that
$\left( M_{P_{LG}} \right)^{\sum w_i}\iso (-)[2]$,
where $M_{P_{LG}}=\ms{T}_{\O_X} \comp \ms{L}_{\O_X(1)}$ in the notation of Sec.~\ref{s:mg}.
The proof constructs a Beilinson's resolution of the diagonal for  $\P^n_{w_0\ldots w_n}$ using the full exceptional collection
$$
\O,\O(1),\ldots , \O(\mbox{$\sum$} w_i-1)
$$
and its dual. Canonaco generalized this approach \cite{AC}, and shows that given a  full exceptional collection on a smooth stack, it leads to a Beilinson type resolution. This resolution is then used to prove Prop.~\ref{pullback} and Prop.~\ref{pushout}.

The aim of this paper is to understand the identities stemming from quantum symmetries in the case when $\c M_K(X)$ is higher dimensional. The recurring feature in our investigations will be the existence of  exceptional collections, either on  divisors in the \CY , or in the ambient space.

Exceptional collections have appeared in the physics literature before \cite{Zaslow:1994nk,Hori:2000ck}.
They were  first applied in the context of Landau-Ginzburg models in \cite{Govindarajan:2000vi,Tomasiello:2000ym,Mayr:2000as}.
They reappeared in the AdS/CFT literature \cite{Cachazo:2001sg,Wijnholt:2002qz} in the context of quiver gauge theories \cite{Douglas:Moore}. Their role in determining the gauge theory living on D-branes placed at \CY\ singularities was clarified in \cite{Herzog:2005sy,en:Chris2}. In this paper we show that they can also be used to establish the monodromy identities.

The organization of the paper is as follows. In Section~2 we briefly review some of the Fourier-Mukai technology, which is used throughout the paper. In Section~3 we consider three very different examples with two dimensional moduli space, and show that in every instance, the monodromy identities that follow from the emergence of additional cyclic symmetries in moduli space {\em always} follow from  Prop.~\ref{pullback} and Prop.~\ref{pushout}. We conclude the paper with a discussion of how general our results are, and an outlook to possible generalizations.

\section{ Monodromies as autoequivalences}    \label{s:m}

We start this section with a brief review of Fourier-Mukai functors. Then we  express the various  monodromy actions on D-branes in terms of Fourier-Mukai equivalences.

\subsection{Fourier-Mukai functors}\label{s:fmf1}

For the convenience of the reader we review some of the key notions concerning Fourier-Mukai functors, and at same time specify the conventions used. We will make extensive use of this technology in the rest of the paper.  Our notation follows \cite{en:fracC3}.  

Given two non-singular proper algebraic varieties (or smooth Deligne-Mumford stacks), $X_1$ and $X_2$, an object ${\cal K} \in \D(X_1\! \times\! X_2)$ determines a functor of triangulated categories $\Phi_{\cal K}\!: \D(X_1) \to \D(X_2)$ by the formula\footnote{
$\R p_{2*}$ is the total right derived functor of $p_{2*}$, i.e., it is an exact functor from $\D(X)$ to $\D(X)$. Similarly, $\Ltensor$ is the total left derived functor of $\otimes$. Most of the time these decorations will be omitted.}
\begin{equation*}
\Phi_{\cal K}(A):=\R p_{2*} \big(\,{\cal K} \Ltensor p_1^*(A)\,\big)\,,
\end{equation*}
where $p_i\!: X\! \times\! X \to X$ is projection to the $i$th factor:
\begin{equation*}
\xymatrix{
  &X_1 \!\times\! X_2\ar[dl]_{p_1}\ar[dr]^{p_2}&\\
  X_1 & & X_2\,.}
\end{equation*}
The object ${\cal K} \in \D(X_1\!\times\! X_2)$ is called the {\bf kernel} of the Fourier-Mukai functor $\Phi_{\cal K}$.

It is convenient to introduce the {\bf external tensor product} of two objects $A\in\D(X_1)$ and $B\in\D(X_2)$ by the formula
\begin{equation*}
A\boxtimes B=p_2^*A\Ltensor p_1^*B\,.
\end{equation*}

The importance of Fourier-Mukai functors when dealing with derived categories stems from the following theorem of Orlov (Theorem 2.18 in \cite{Orlov:96}), later generalized for smooth quotient stacks associated to normal projective varieties \cite{Kawamata:DC})
\begin{thm} \label{thm:orlov}
Let $X_1$ and $X_2$ be smooth projective varieties.
Suppose that $\mathsf{F}\!: \D(X_1)\to\D(X_2)$ is an equivalence of triangulated categories. Then there exists an object $\c K\in \D(X_1\!\times \! X_2)$, unique up to isomorphism, such that the functors $\mathsf{F}$ and $\Phi_{\c K}$ are isomorphic.
\end{thm}

The first question to ask is how to compose Fourier-Mukai (FM) functors. Accordingly, let $X_1$ $X_2$ and $X_3$ be three non-singular varieties, while let ${\cal F} \in \D(X_1\! \times\! X_2)$ and ${\cal G} \in \D(X_2\! \times\! X_3 )$ be two kernels. Let $p_{i j}\colon X_1\! \times\! X_2\! \times\! X_3\to X_i\! \times\! X_j$ be the projection map. A well-known fact is the following:
\begin{prop}\label{prop1}
The composition of the functors $\Phi_{\cal F}$ and $\Phi_{\cal G}$  is given by the formula
\begin{equation*}
\Phi_{\cal G}\comp \Phi_{\cal F} \simeq\Phi_{\cal H}\,,\quad {\rm where}\quad
{\cal H}=\R p_{13*} \big(\, p_{23}^* ( {\cal G})\Ltensor  p_{12}^* ({\cal F})\big)\,.
\end{equation*}
\end{prop}

Prop.~\ref{prop1} shows that  composing two FM functors gives another FM functor, with a simple kernel.

Now we have all the technical tools ready to study the monodromy actions of physical interest.

\subsection{Monodromies in general}\label{s:mg}

As discussed in the introduction, the moduli space of CFT's contains the moduli space of Ricci-flat Kahler metrics. This, in turn, at least locally has a product structure, with the moduli space of Kahler forms being one of the factors. This is the moduli space of interest to us.
This space is a priori non-compact, and its compactification consists of two different types of boundary divisors. First we have the {\em large volume} divisors.  These correspond to certain cycles being given infinite volume. The second type of boundary divisors are the irreducible components of the {\em discriminant}. In this case the CFT becomes singular. Generically this happens because some D-brane (or several of them, even infinitely many) becomes massless at that point, and therefore the effective CFT description breaks down. For the quintic this breakdown happens at the well known conifold point.

The monodromy actions around the above divisors are  understood to some extent. An extensive treatment of monodromies in terms of Fourier-Mukai functors was given in \cite{en:Horja}. We will review now what is known.

Large volume monodromies are shifts in the $B$ field: ``$B\mapsto B+1$''. If the Kahler cone is higher dimensional, then we need to be more precise, and specify a two-form, or equivalently a divisor $D$. Then the monodromy becomes $B\mapsto B+D$. We will have more to say about the specific $D$'s soon. 

The simplest physical effect of this monodromy on a D-brane is to shift its charge, and this translates in the Chan-Paton language into tensoring with the line bundle $\O_X(D)$. This observation readily extends to the derived category:
 \begin{prop}\label{p:lr} 
The large radius monodromy associated to the divisor $D$ is
\begin{equation*}
   \ms{L}_{D}(\mathsf{B}) = \mathsf{B}\Ltensor \O_X(D)\,,\qquad \mbox{for all $\mathsf{B} \in \D(X)$}\,.
\end{equation*}
Furthermore, this is a Fourier-Mukai functor $\Phi_{{\cal L}}$, with kernel 
\begin{equation*}
{\cal L}=\delta_*\O_X(D)\,,
\end{equation*}
where  $\delta \!: X \hookrightarrow X\! \times\! X$ is  the diagonal embedding.
\end{prop}

Now we turn our attention to the conifold-type monodromies. For this we need to introduce the Fourier-Mukai functor with kernel $\Cone{ \mathsf{A}^{\roof}\boxtimes\,\mathsf{A} \to \O_\Delta}$, where $\O_\Delta = \delta_*\O_X$, and  for $\mathsf{A}\in\D(X)$ its derived dual is
\begin{equation*}
\mathsf{A}^{\roof}= \R\!\!\sHom_{\D(X)}( \mathsf{A}, \O_X).
\end{equation*}
By Lemma~3.2 of \cite{ST:braid}, for any $\mathsf{B}\in\D(X)$:
\begin{equation*}
\Phi_{ \Cone{ \mathsf{A}^{\roof}\boxtimes\,\mathsf{A} \to \O_\Delta} }(\mathsf{B}) \iso
\Cone{ \Hom_{\D(X)}(\mathsf{A},\mathsf{B})\Ltensor\mathsf{A}\longrightarrow \mathsf{B} }.
\end{equation*}

Since the functor $\Phi_{ \Cone{ \mathsf{A}^{\roof}\boxtimes\,\mathsf{A}\to\O_\Delta } }$ will play a crucial role, we introduce a  notation for it:
\begin{equation}\label{e:refl}
\ms{T}_{\mathsf{A}} :=
\Phi_{ \Cone{\mathsf{A}^{\roof}\boxtimes\,\mathsf{A}\to\O_\Delta } }\,,\quad
\ms{T}_{\mathsf{A}}(\mathsf{B}) =
\Cone{ \Hom_{\D(X)}(\mathsf{A},\mathsf{B})\Ltensor\mathsf{A}\longrightarrow \mathsf{B} }\,.
\end{equation}
The functor $\ms{T}_{\mathsf{A}}$ is sometimes referred to as the Seidel-Thomas twist functor.

Returning to conifold-type monodromies, we have the following conjecture from \cite{en:Horja}:
\begin{conj}\label{conj:a}
If we loop around a component of the discriminant locus associated with a single D-brane $\mathsf{A}$ becoming
massless, then this results in a relabeling of D-branes  by applying $\ms{T}_{\mathsf{A}}$.
\end{conj}
 
The question of when is $\ms{T}_{\mathsf{A}}$ an autoequivalence has a simple answer. For this we need the following definition:
\begin{definition} \label{def:spherical}
Let $X$ be smooth projective \CY\ variety (stack) of dimension $n$. An object $\mathsf{E}$ in $\D(X)$ is called {\em n-spherical} if $\Ext^r_{\D(X)}(\mathsf{E},\, \mathsf{E})\cong \H^r(S^n,\,\C)$, that is $\C$ for $r = 0,n$ and zero otherwise.
\end{definition}

One of the main results of \cite{ST:braid} is the following:
\begin{thm}({\rm Prop. 2.10 in \cite{ST:braid})}
If the object $\mathsf{E}\in \D(X)$ is n-spherical, then the functor $\ms{T}_{\mathsf{E}}$ is an autoequivalence.
\end{thm}

Let us mention at this point that in the rest of the paper whenever we  have an expression involving  the functor $\ms{T}_{\mathsf{E}}$, then $\mathsf{E}$ will always be spherical.

\subsection{Exceptional collections and autoequivalences}    \label{s:EC}

Exceptional collections play a surprisingly central role concerning the monodromy identities that we are to discuss. First we recall some facts about them, then we will list two propositions from \cite{AC} which constitute the technical backbone of this paper.

\begin{definition}
Consider the bounded derived category of coherent sheaves $\D(X)$ on the algebraic variety (smooth stack) $X$.
\begin{enumerate}
\item An object ${\c E} \in \D(X)$ is called {\em exceptional} if $\Ext^q({\c E},{\c E}) = 0$ for $q\neq 0$ and $\Ext^0({\c E},{\c E}) = {\mathbb C}$.
\item An {\em exceptional collection} $ ({\c E}_1, {\c E}_2, \ldots, {\c E}_n)$ in $\D(X)$ is an ordered collection of exceptional objects such that
\[
\Ext^q({\c E}_i, {\c E}_j) = 0 ,\qquad \mbox{for all q, whenever} \; i > j \ .
\]
\item An exceptional collection is {\em complete} or {\em full} if it generates $\D(X)$.
\end{enumerate}
\end{definition}

The existence of a full and strong exceptional collection for a given variety $X$ constrains the structure of $X$ considerably.  In particular, no smooth projective (and therefore compact) \CY\ variety admits such a collection; the obstruction comes from Serre duality.  The exceptional collections we are interested in are constructed on an exceptional divisor or the ambient space, i.e., where the \CY\ is embedded, rather than on $X$ itself.

Now we  turn to two propositions that  will be used repeatedly in the remainder of this paper. As we will see, in all the examples  considered, the monodromy identities dictated by CFT can always be explained using these two statements. It is also gratifying to remark that these statements were motivated by precisely the  kind of identities that they are now used to prove. More precisely, these statements were formulated by Alberto Canonaco as a consequence of trying to prove with the author the conjectures substantiated in \cite{en:fracC2} and \cite{en:fracC3}.

Let $X$ and $Y$ be smooth varieties/stacks, with canonical bundles $\omega_X$ resp. $\omega_Y$. Also assume that we have a full exceptional collection in $\D(Y)\colon {\c E}_1, {\c E}_2, \ldots, {\c E}_m$. There are two cases to consider. The first is when $X$ is a \CY\ hypersurface in $Y$. In this case we have that

\begin{prop}\label{pullback}
If $i\colon X\mono Y$ is the inclusion of a hypersurface such that
$\omega_Y\iso\O_Y(-X)$, then
$$\T{i^*{\c E}_{0}}\comp\cdots\comp\T{i^*{\c E}_{m}}\iso\L{\O_X(-X)[2]}.$$
\end{prop}

In the second case $Y$ is a hypersurface in $X$, where $X$ is \CY ; and we have and analogous statement

\begin{prop}\label{pushout}
If $j\colon Y\mono X$ is the inclusion of a hypersurface such that
$j^*\omega_X\iso\O_Y$, then
$$\T{j_*{\c E}_{0}}\comp\cdots\comp\T{j_*{\c E}_{m}}
\iso\L{\O_X(Y)}.$$
\end{prop}

Note that if $X$ is \CY , then the condition $j^*\omega_X\iso\O_Y$ is automatically satisfied.

\section{Examples}    \label{s:Ex}

\subsection{\texorpdfstring{$\C^{\, 2}/\Z_3$}{C2/Z3} }    \label{s:c2z3}

In this section we focus on the $\C^2/\Z_3$ geometric orbifold and the associated CFT. In this case there is only one supersymmetric $\Z_3$ action
\begin{equation*}
(z_1,z_2)\mapsto (\omega z_1,\omega^2 z_2)\,, \qquad \omega ^3=1\,.
\end{equation*}
First let us review some of the findings and notations of \cite{en:fracC2}, then generalize them. The crepant resolution of the singularity, denoted by $X$, has a reducible exceptional divisor. The toric fan of the resolved space $X$ is shown in Fig.~\ref{f:fan}.
\begin{figure}[h]
\begin{equation}\nonumber
\begin{xy} <1.3cm,0cm>:
{\ar 0;(0,1) }, (0,1.2)*\txt{$v_2$=(0,1)}
,{\ar (0,0);(1,0) *+!LD{v_3=\frac{v_1+2v_2}{3}=(1,0)}}
,{\ar 0;(2,-1) *+!LD{v_4=\frac{2v_1+v_2}{3}=(2,-1)}}
,{\ar 0;(3,-2) *+!LD{v_1=(3,-2)}}
,{\ar@{-}@{.>} (0,1);(3,-2) }
\end{xy}
\end{equation}
  \caption{The toric fan for the resolution of the $\C^2/\Z_3$ singularity.}
  \label{f:fan}
\end{figure}
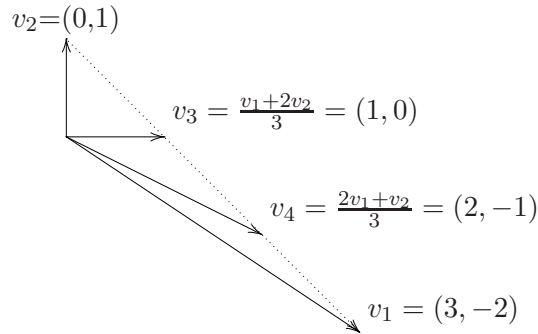
Let $C_i$ denote the divisor associated to the vertex $v_i$, which in this case is a curve. The exceptional locus of the blow-up consists of the divisors $C_3$ and $C_4$, both $-2$ curves.

The curves $C_3$ and $C_4$ are the generators of the Mori cone of effective curves. The Kahler cone is dual to the Mori cone, and both are two dimensional. The Poincare duals of the curves $C_i$ are denoted by  $D_i$. Since we are in two complex dimensions, an irreducible divisor is a curve. This leads to potential confusion. To avoid it, the reader should remember that $C_i$ lives in the second homology $\H_2(X,\Z)$, while $D_i$ lives in the second cohomology $\H^2(X,\Z)$. 

Now let us look at the moduli space of complexified Kahler forms. The point-set spanned by the rays of the toric fan, $\c A=\{v_1,\ldots, v_4 \}$, admits four triangulations, i.e., in the language of the gauged linear sigma model we have four phases. The secondary fan is depicted in Fig.~\ref{fig:Z3}, which is the toric fan of the Kahler moduli space. The four phases are the completely resolved smooth phase; the two phases where one of the $\P^1$'s has been blown up to partially resolve the $\Z_3$ fixed point to a $\Z_2$ fixed point; and finally the $\Z_3$ orbifold phase.

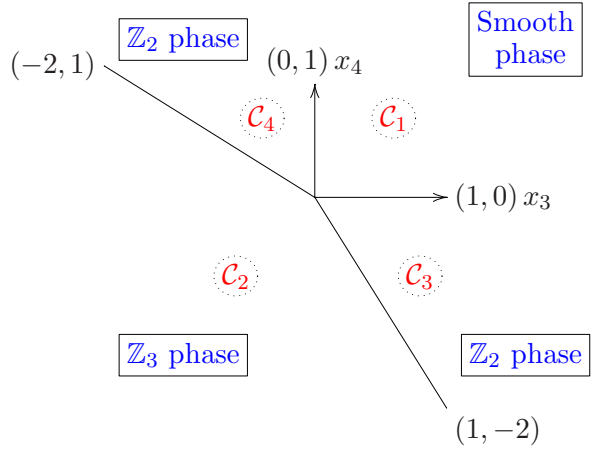
\begin{figure}
\begin{equation}\nonumber
\begin{xy} <3.5cm,0cm>:
{\ar (0,0);(.5,0) *+!L{(1,0)\,x_3}}
,{\ar 0;(0,.6) *+!U{(0,1)\,x_4}}
,{\ar@{-} 0;(.5,-.8) *+!LU{(1,-2)}}
,{\ar@{-} 0;(-.8,.5) *+!R{(-2,1)}}
,(.8,.6)*+[F]\txt{\color{blue} Smooth\\ \ \color{blue} phase},
,(.3,.3)*+[o][F.]{\color{red} {\mathcal C}_1}
,(-.5,-.6)*+[F]\txt{\color{blue} $\Z_3$  \color{blue} phase}
,(-.3,-.3)*+[o][F.]{\color{red} {\mathcal C}_2}
,(.8,-.6)*+[F]\txt{\color{blue} $\Z_2$  \color{blue} phase}
,(.4,-.3)*+[o][F.]{\color{red} {\mathcal C}_3}
,(-.5,.6)*+[F]\txt{\color{blue} $\Z_2$  \color{blue} phase}
,(-.2,.3)*+[o][F.]{\color{red} {\mathcal C}_4}
\end{xy}
\end{equation}
  \caption{The phase structure of the $\C^2/\Z_3$ model.}
  \label{fig:Z3}
\end{figure}

The orbifold points in the moduli space are themselves singular points. This fact is related to the quantum symmetry of an orbifold theory. For either of the $\Z_2$ points, one has a $\C^2/\Z_2$ singularity with weights $(1,-1)$, while the $\Z_3$ point the moduli space locally is of the form $\C^2/\Z_3$, with weights $(1,2)$. 

The four maximal cones of Fig.~\ref{fig:Z3}, ${\mathcal C}_1$,\ldots , ${\mathcal C}_4$, correspond to the four distinguished phase points. The four edges correspond to curves in the moduli space, denoted ${\mathcal L}_1,\ldots , {\mathcal L}_4$. These are all weighted projective lines, all isomorphic to $\P^1$ as varieties, but not as stacks. The four curves connecting the different phase points are sketched in Fig.~\ref{fig:modsp}, together with the discriminant locus of singular CFT's. The discriminant $\Delta_0$ intersects the four lines transversely. We depicted this fact in Fig.~\ref{fig:modsp} using short segments.

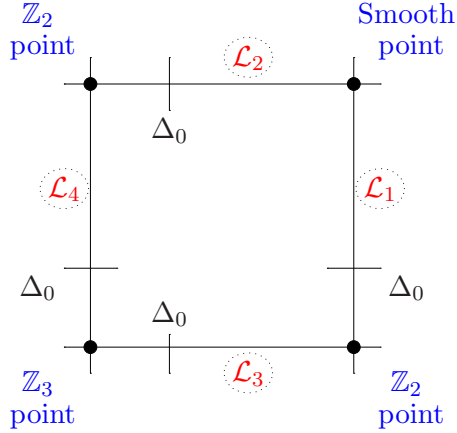
\begin{figure}
\begin{equation}\nonumber
\begin{xy} <3.5cm,0cm>:
{\ar@{-} (-.1,0);(1.1,0) }	,{\ar@{-} (1,.1);(1,-1.1) }	,{\ar@{-} (1.1,-1);(-.1,-1)}	,{\ar@{-} (0,.1);(0,-1.1)} 
,(1.2,.2)*\txt{\color{blue} Smooth\\ \ \color{blue} point}
,(.6,.1)*+[o][F.]{\color{red} {\mathcal L}_2}
,(1.2,-1.2)*\txt{\color{blue} $\Z_2$\\ \ \color{blue} point}
,(1.1,-.4)*+[o][F.]{\color{red} {\mathcal L}_1}
,(-.2,-1.2)*\txt{\color{blue} $\Z_3$\\ \ \color{blue} point}
,(.6,-1.1)*+[o][F.]{\color{red} {\mathcal L}_3}
,(-.2,.2)*\txt{\color{blue} $\Z_2$\\ \ \color{blue} point}
,(-.1,-.4)*+[o][F.]{\color{red} {\mathcal L}_4}
,(0,0)*+=[o]=<1.7mm>[F**:black][black]{.},(0,-1)*+=[o]=<1.7mm>[F**:black][black]{.}
,(1,0)*+=[o]=<1.7mm>[F**:black][black]{.},(1,-1)*+=[o]=<1.7mm>[F**:black][black]{.}
,{\ar@{-} (0.3,0.1);(0.3,-.1)   *+!U{\Delta_0}}		,{\ar@{-} (0.3,-1.1);(0.3,-.8)   *+!U{\Delta_0}}			
,{\ar@{-} (0.9,-.7);(1.2,-.7)   *+!U{\Delta_0}}		,{\ar@{-} (.1,-.7);(-.2,-.7)   *+!U{\Delta_0}}	
\end{xy}
\end{equation}
  \caption{The moduli space of the $\C^2/\Z_3$ model.}
  \label{fig:modsp}
\end{figure}

When talking about monodromy there are two cases to be considered. One can loop around a divisor, i.e., real codimension two objects; or one can loop around a point inside a complex curve. Of course the two notions are not unrelated. Our interest will be the second type of monodromy: looping around a point inside a curve.\footnote{As explained by many authors, since ${\mathcal L}_i$ is not part of the moduli space, rather it is only a compactification divisor, one cannot consider loops inside it. Instead, the loops in question are infinitesimally close to being in ${\mathcal L}_i$. }

It was shown in \cite{en:fracC2} that monodromy around the $\Z_2$ point inside ${\mathcal L}_1$ is
\begin{equation}\label{e:m1} 
\ms{M}_{\Z_2}\, = \, \ms{T}_{i_*\O_{C_3}} \comp \, \ms{L}_{D_2}\,,
\end{equation}
while monodromy around the $\Z_2$ point inside ${\mathcal L}_2$ is
\begin{equation}\label{e:m11} 
\ms{M}_{\Z_2}'\, = \, \ms{T}_{i_*\O_{C_4}} \comp \, \ms{L}_{D_1}\,.
\end{equation}
Here $i\colon C_3\mono X$ resp. $j\colon C_4\mono X$ are the embedding maps, while $C_3$ and $C_4$ are the exceptional divisors defined at the beginning of the subsection. Finally,  monodromy inside ${\mathcal L}_3$ around the $\Z_3$ point is given by
\begin{equation}\label{e:m2}
\ms{M}_{\Z_3}\, = \, \ms{T}_{j_*\O_{C_4}}\comp \, \ms{M}_{\Z_2}
\, = \, \ms{T}_{j_*\O_{C_4}}\comp \,  \ms{T}_{i_*\O_{C_3}} \comp \, \ms{L}_{D_2}\,.
\end{equation}

Using an approach analogous to the one deployed in \cite{en:Alberto}, \cite{en:fracC2} showed that
\begin{equation}\label{c21}
 (\ms{M}_{\Z_2})^2\, = \L{D_1}\,\quad (\ms{M}_{\Z_2}')^2\, =\, \L{D_2}.
\end{equation}
These identities were reproved in \cite{AC} using Proposition~\ref{pullback}.

\cite{en:fracC2}  also conjectured that 
\begin{equation}\label{c22}
(\ms{M}_{\Z_3})^{\, 3}\, \cong\, \id_{\D(X)},
\end{equation} 
and checked this statement at the level of Chern characters ($\id_{\D(X)}$ is the identity functor of ${\D(X)}$). Here we provide a proof for (\ref{c22}), which is different from the one in \cite{AC} by being more closely tied with the physics of the example, and which will pave the way for  similar result presented in the sequel, which of course are not proven in \cite{AC}.

We start by recalling Lemma~8.21 from \cite{Huybrechts}, which shows how to conjugate a Seidel-Thomas twist functor by an autoequivalence.
\begin{lemma}\label{STcomp}
If $\s{F}\in\D(X)$ and ${\mathsf G}$ is an autoequivalence
of $\D(X)$, then
$${\mathsf G}\comp\T{\s{F}}\iso
\T{{\mathsf G}(\s{F})}\comp{\mathsf G}.$$
\end{lemma}

First we simplify $(\ms{M}_{\Z_3})^{2}$ using lemma~\ref{STcomp}:\footnote{For brevity we omit $i$ and $j$ from $\ms{T}_{i_*\O_{C_3}} $ and $\ms{T}_{j_*\O_{C_4}}$.}
\begin{equation*}
 \ms{M}_{\Z_3}^{\, 2} = \ms{M}_{\Z_3} \comp (\T{\O_{C_4}} \comp \ms{M}_{\Z_2})
\iso  \T{\ms{M}_{\Z_3}(\O_{C_4})}  \comp  \ms{M}_{\Z_3} \comp \ms{M}_{\Z_2}
=  \T{\ms{M}_{\Z_3}(\O_{C_4})}  \comp   \T{\O_{C_4}} \comp \ms{M}_{\Z_2}^{\, 2}.
\end{equation*}
As it was shown in Sec. 4.2 of \cite{en:fracC2}, $\ms{M}_{\Z_3}(\O_{C_4})= {\O_{C_3}}$; while (\ref{c21}) shows that $\ms{M}_{\Z_2}^{\, 2}\iso \L{D_1}$, therefore
\begin{equation}\label{1}
 \ms{M}_{\Z_3}^{\, 2} \iso \T{\O_{C_3}} \comp \T{\O_{C_4}}  \comp \L{D_1}
=\T{\O_{C_3}} \comp \ms{M}_{\Z_2}',
\end{equation}
where $\ms{M}_{\Z_2}'$ was defined in (\ref{e:m11}).

Now observe that 
\begin{equation}\label{2}
\ms{M}_{\Z_3}= \T{\O_{C_4}} \comp \ms{M}_{\Z_2}=\ms{M}_{\Z_2}'  \comp \L{D_1}^{-1}\comp \ms{M}_{\Z_2}.
\end{equation}
Therefore (\ref{1}) and (\ref{2}) imply that
\begin{equation*}
 \ms{M}_{\Z_3}^{\, 3} \iso \T{\O_{C_3}} \comp (\ms{M}_{\Z_2}')^{\, 2} \comp \L{D_1}^{-1}\comp \ms{M}_{\Z_2}.
\end{equation*}
But $(\ms{M}_{\Z_2}')^{\, 2}\iso \L{D_2}$ by (\ref{c21}), while  $\T{\O_{C_3}} \comp \L{D_2}=\ms{M}_{\Z_2}$. Thus
\begin{equation*}
 \ms{M}_{\Z_3}^{\, 3} \iso \ms{M}_{\Z_2} \comp \L{D_1}^{-1}\comp \ms{M}_{\Z_2}\iso \id_{\D(X)},
\end{equation*}
where in the last relation we used (\ref{c21}) again.

The idea that one has to take away from this proof is the concept, rather than the manipulations. One starts out with $\ms{M}_{\Z_3}$, and then rewrite it's third power in such a way that the already established identities in (\ref{c21}) can be used. Note also that both identities in (\ref{c21}) were needed, along with two of the three fractional branes, $\O_{C_3}$ and $\O_{C_4}$, and the fact that the $\Z_3$ monodromy permutes the fractional branes, i.e., $\ms{M}_{\Z_3}(\O_{C_4})= {\O_{C_3}}$ from \cite{en:fracC2}. The same philosophy will guide the proof of the next subsection.

\subsection{\texorpdfstring{$\C^{\, 3}/\Z_5$}{C3/Z5}}    \label{s:c2z5}

Let us now turn to the example of $\C^{\, 3}/\Z_5$, as treated in \cite{en:fracC3}. First we review the relevant toric geometry of $\C^3/\Z_5$, then  the moduli space of complexified Kahler forms. Throughout, we follow closely the notation of  \cite{en:fracC3}.

Once again, there is a unique supersymmetric $\Z_5$ action, i.e., $\Z_5\subset \operatorname{SL}(3,\Z)$:
\begin{equation*}
(z_1,z_2,z_3)\mapsto (\omega z_1,\omega z_2,\omega^3 z_3)\,, \qquad \omega ^5=1\,.
\end{equation*}
The toric fan of the resolved space $X$ is the cone over Fig.~\ref{f:fan1}. We denote the divisor associated to $v_i$ by $D_i$. The exceptional locus of the blow-up is reducible, with two irreducible components corresponding to $v_4$ and $v_5$.

\begin{figure}[h] 
\begin{equation}\nonumber
\begin{xy} <1.3cm,0cm>:
0*{\dir{*}}*++!D{\color{blue}v_2}="2",  (-1,-2.58)*{\dir{*}}*+!RU{\color{blue}v_1}="1"	,(1,-2.58)*{\dir{*}}*+!LU{\color{blue}v_3}="3"
,(0,-1)*{\dir{*}}="5"	,(0,-2)*{\dir{*}}="4"
\ar@{-} "1"; 0	\ar@{-} "3"; 0	\ar@{-} "1"; "3"	\ar@{-} "2";"4" 	
\ar@{-} "3"; "4"	\ar@{-}"3"; "5"	 \ar@{-} "1"; "4"	\ar@{-}"1"; "5"  	
,\POS"5"*+!L{\color{red}v_5}	,\POS"4"*++!L{\color{red}v_4}
\end{xy}
\end{equation}
  \caption{The toric fan for the resolution of the $\C^3/\Z_5$ singularity.}
  \label{f:fan1}
\end{figure}
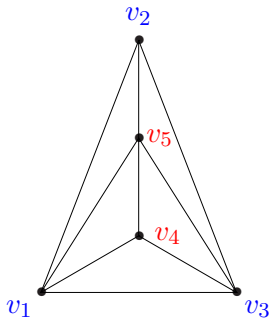

Toric geometry immediately tells us that the divisor $D_4$ is a $\P^2$, while $D_5$ is the Hirzebruch surface $\F_3$. Let $f$ be the fiber of $\F_3$ (e.g., the cone generated by $v_1$ and $v_5$), and $h$ its  $-3$ section (the cone generated by $v_4$ and $v_5$). At the same time, $h$ is the hyperplane class of $\P^2$, while $f$ does not intersect $\P^2$. The curves $h$ and $f$ are the generators of the Mori cone of effective curves. Shrinking  $h$ also shrinks the divisors $D_4$, and hence gives a Type~II contraction, while shrinking  $f$ collapses the Hirzebruch surface $\F_3$ onto its base, giving a Type~III degeneration.

\begin{figure}[h]
\begin{equation}\nonumber
\begin{xy} <3.5cm,0cm>:
{\ar (0,0);(.5,0) *+!L{(1,0)\,x_1}}
,{\ar 0;(0,.6) *+!U{(0,1)\,x_2}}
,{\ar@{-} 0;(.5,-.8) *+!LU{(1,-2)}}
,{\ar@{-} 0;(-1.2,.5) *+!R{(-3,1)}}
,(.8,.6)*+[F]\txt{\color{blue} Smooth\\ \ \color{blue} phase},
,(.3,.3)*+[o][F.]{\color{red} {\mathcal C}_1}
,(-.5,-.6)*+[F]\txt{\color{blue} $\Z_5$  phase}
,(-.3,-.3)*+[o][F.]{\color{red} {\mathcal C}_4}
,(.8,-.6)*+[F]\txt{\color{blue} $\Z_2$  phase}
,(.4,-.3)*+[o][F.]{\color{red} {\mathcal C}_2}
,(-.6,.6)*+[F]\txt{\color{blue} $\Z_3$  phase}
,(-.2,.3)*+[o][F.]{\color{red} {\mathcal C}_3}
\end{xy}
\end{equation}
  \caption{The phase structure of the $\C^3/\Z_5$ model.}
  \label{fig:Z5}
\end{figure}
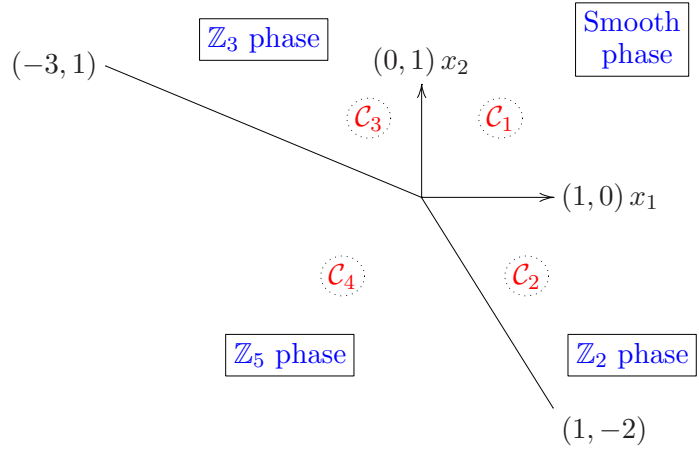

There are again four phases. The secondary fan is depicted in Fig.~\ref{fig:Z5}. The four phases are as follows: the completely resolved smooth phase; the two phases where one of the compact divisors $D_4$ or  $D_5$ has been blown up to partially resolve the $\Z_5$ fixed point; and finally the $\Z_5$ orbifold phase.

The phase corresponding to the cone ${\mathcal C}_2$ can be reached from the smooth phase ${\mathcal C}_1$ by blowing down the divisor $D_5$. This creates a line of $\Z_2$ singularities in the \CY . We will refer to this phase as {\em the $\Z_2$ phase}. Similarly, the phase ${\mathcal C}_3$ is reached by blowing down the divisor $D_4$, and creates a $\C^3/\Z_3$ singularity. We call this {\em the $\Z_3$ phase}.

The orbifold points in the moduli space are themselves singular points. The $\Z_2$ point is a  $\C^2/\Z_2$ singularity with weights $(1,-1)$, while the $\Z_3$ point the moduli space locally is of the form $\C^2/\Z_3$, with weights $(1,2)$.

We have already seen that the four maximal cones ${\mathcal C}_1$,\ldots , ${\mathcal C}_4$ in Fig.~\ref{fig:Z5} correspond to the four distinguished phase points. Similarly, the four rays correspond to curves in the moduli space, and once again are weighted projective lines: ${\mathcal L}_1,\ldots ,{\mathcal L}_4$.

\begin{figure}[h]
\begin{equation}\nonumber
\begin{xy} <3.5cm,0cm>:
{\ar@{-} (-.1,0);(1.1,0) }	,{\ar@{-} (1,.1);(1,-1.1) }	,{\ar@{-} (1.1,-1);(-.1,-1)}	,{\ar@{-} (0,.1);(0,-1.1)} 
,(1.2,.2)*\txt{\color{blue} Smooth\\ \ \color{blue} point}
,(.6,.1)*+[o][F.]{\color{red} {\mathcal L}_2}
,(1.2,-1.2)*\txt{\color{blue} $\Z_2$\\ \ \color{blue} point}
,(1.1,-.3)*+[o][F.]{\color{red} {\mathcal L}_1}
,(-.2,-1.2)*\txt{\color{blue} $\Z_5$\\ \ \color{blue} point}
,(.6,-1.1)*+[o][F.]{\color{red} {\mathcal L}_3}
,(-.2,.2)*\txt{\color{blue} $\Z_3$\\ \ \color{blue} point}
,(-.1,-.4)*+[o][F.]{\color{red} {\mathcal L}_4}
,{\ar@{-} (0.3,0.1);(0.3,-.1)   *+!U{\Delta_0}}		,{\ar@{-} (0.3,-1.1);(0.3,-.8)   *+!U{\Delta_0}}			
,{\ar@{-} (.1,-.7);(-.2,-.7)   *+!U{\Delta_0}}		,(0.7,-.6)   *+!U{\Delta_0}
,{(.8,-.35);(.8,-.65) **\crv{(1.091,-.45) & (1.05,-.55)}}
,(0,0)*{\bullet}, (0,-1)*{\bullet} 	,(1,0)*{\bullet} ,(1,-1)*{\bullet} 
	,(1,-.5)*{\bullet}
\end{xy}
\end{equation}
  \caption{The moduli space of the $\C^3/\Z_5$ model.}
  \label{fig:modsp5}
\end{figure}
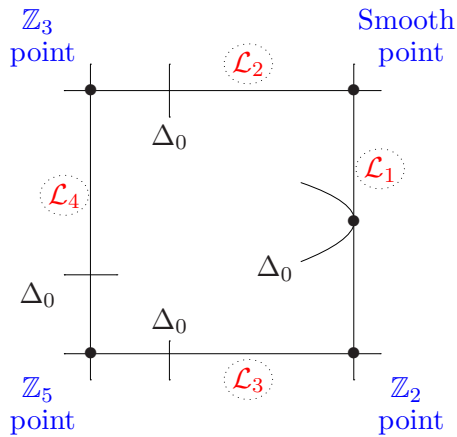

The discriminant intersects ${\mathcal L}_1$ tangentially, while it is transverse to the other ${\mathcal L}_i$'s'. We depicted this fact in Fig.~\ref{fig:modsp5} using a parabola and resp.  short segments.

It was shown in \cite{en:fracC3} that monodromy around the $\Z_2$ point inside ${\mathcal L}_1$ is \footnote{Once again, $i$ and $j$ are the embeddings, and will be omitted later on.}
\begin{equation}\label{e:m15} 
\ms{M}_{\Z_2}\, = \,
\ms{T}_{j_*\O_{D_5}(-f)} \, \comp \, \ms{T}_{j_*\O_{D_5}}  \comp \, \ms{L}_{D_2}\,,
\end{equation}
and monodromy around the $\Z_3$ point inside ${\mathcal L}_2$ is
\begin{equation}\label{e:m1z} 
\ms{M}_{\Z_3}\, = \, \ms{T}_{i_*\O_{D_4}} \comp \, \ms{L}_{D_1}  \,,
\end{equation}
while monodromy inside ${\mathcal L}_3$ around the $\Z_5$ point is given by 
\begin{equation}
\ms{M}_{\Z_5}\, = \, \ms{T}_{i_*\O_{D_4}}\comp \, \ms{M}_{\Z_2}\,.
\end{equation}

In  \cite{en:fracC3} it was conjectured that 
\begin{equation}\label{c31}
 (\ms{M}_{\Z_2})^2\, = \L{D_1}\,\quad (\ms{M}_{\Z_3})^3\, =\, \L{D_2},
\end{equation}
and the statements were checked  at the level of Chern characters. These identities were proved in \cite{AC} using Proposition~\ref{pullback}, and full exceptional collections on $\F_3$ resp. $\P^2$.  \cite{en:fracC3}  also conjectured that
\begin{prop}\label{c32}
$$\ms{M}_{\Z_5}^{\, 5}\, \cong\, \id_{\D(X)}.$$
\end{prop}

\begin{proof}
We prove the statement using (\ref{c31}), and  the philosophy of the previous subsection. For this we first rewrite  $\ms{M}_{\Z_5}^{\, 2}$ in a more convenient form using Lemma~\ref{STcomp}:
\begin{equation}
 \ms{M}_{\Z_5}^{\, 2} = \ms{M}_{\Z_5} \comp (\T{\O_{D_4}} \comp \ms{M}_{\Z_2})
\iso   \T{\ms{M}_{\Z_5}(\O_{D_4})}\comp\ms{M}_{\Z_5} \comp  \ms{M}_{\Z_2}=\T{\ms{M}_{\Z_5}(\O_{D_4})}\comp\T{\O_{D_4}}\comp \ms{M}_{\Z_2}^{\, 2}.
\end{equation}
But $\ms{M}_{\Z_2}^{\, 2}\iso\L{D_1}$ by (\ref{c31}), while $\T{\O_{D_4}}  \comp \L{D_1}=\ms{M}_{\Z_3}$. Furthermore, Sec. 4.1.3 of \cite{en:fracC3} proves that $\ms{M}_{\Z_5}(\O_{D_4})=\O_{D_5}(f)$. (As shown in \cite{en:fracC3}, both are fractional branes, and this is just the statement that the $\Z_5$ monodromy permutes the fractional branes.) Therefore
\begin{equation}\label{3}
 \ms{M}_{\Z_5}^{\, 2} \iso \T{\O_{D_5}(f)}\comp \ms{M}_{\Z_3}.
\end{equation}
A short computation shows that $\ms{M}_{\Z_3}(\O_{D_5}(f))=\O_{D_5}(2f)$. Using this fact, (\ref{3}) and  Lemma~\ref{STcomp}, we have that
\begin{equation*}
 \ms{M}_{\Z_5}^{\, 4} \iso \T{\O_{D_5}(f)}\comp \T{\O_{D_5}(2f)}\comp \ms{M}_{\Z_3}^{\, 2}.
\end{equation*}
On the other hand it is easy to verify that
\begin{equation*}
\T{\O_{D_5}(f)}\comp \T{\O_{D_5}(2f)} \iso \L{D_1}^{\, 2}\comp \ms{M}_{\Z_2} \comp\L{D_1}^{-2}\comp\L{D_2}^{-1}.
\end{equation*}
Consequently,
\begin{equation*}
\begin{split}
 \ms{M}_{\Z_5}^{\, 5}
&= (\T{\O_{D_4}} \comp \ms{M}_{\Z_2})\comp
	(\L{D_1}^{\, 2}\comp \ms{M}_{\Z_2} \comp\L{D_1}^{-2}\comp\L{D_2}^{-1})\comp \ms{M}_{\Z_3}^{\, 2}\\
&\iso \T{\O_{D_4}} \comp
	(\ms{M}_{\Z_2}\comp \L{D_1}^{\, 2}\comp \ms{M}_{\Z_2}) \comp (\L{D_1}^{-2}\comp\L{D_2}^{-1})\comp \ms{M}_{\Z_3}^{\, 2}\\
& \iso \T{\O_{D_4}} \comp \L{D_1}^{\, 3}\comp (\L{D_1}^{-2}\comp\L{D_2}^{-1})\comp \ms{M}_{\Z_3}^{\, 2}\\
&= (\T{\O_{D_4}} \comp \L{D_1})\comp\L{D_2}^{-1}\comp \ms{M}_{\Z_3}^{\, 2}
=\ms{M}_{\Z_3}\comp\L{D_2}^{-1}\comp \ms{M}_{\Z_3}^{\, 2}\iso\id_{\D(X)}.
\end{split}
\end{equation*}
To go from line one to line two we used the fact that  composition of Fourier-Mukai functors is associative. To go to line three we used the first relation in (\ref{c31}). In the last line we used (\ref{e:m1z}), and the second relation in (\ref{c31}).
\end{proof}

\subsection{A compact example: \texorpdfstring{$\P^4_{9,6,1,1,1}[18]$}{P4{9,6,1,1,1}[18]}}    \label{s:P4}

The degree $18$ hypersurface in the weighted projective space $\P^4_{9,6,1,1,1}$ is a well studied example of a \CY\ 3-fold with two dimensional complexified Kahler moduli space. We will follow the notation of \cite{navigation}, where the relevant identities that are the subject of this paper were partly checked at K-theory level. We will lift these identities to the derived category, and prove them. This example is particularly interesting, since both Proposition~\ref{pullback} and \ref{pushout} will be needed in the proof, while the examples studied so far used only one of them per example.

First we review the geometry of the blown-up $\P^4_{9,6,1,1,1}$, which we call  $Z$, as the details will be important in the sequel. $Z$  is particularly easy to describe torically: its fan has the same rays, $v_1,\ldots,v_5$, as $\P^4_{9,6,1,1,1}$, and one additional ray corresponding to the blow-up,  $v_6$:
\begin{equation}\label{x2}
 \begin{split}
  &v_1=(1,0,0,0),\quad v_2=(0,1,0,0),\quad v_3=(0,0,1,0)\\
  &v_4=(0,0,0,1),\quad v_5=(-9,-6,-1,-1),\quad v_6=(-3,-2,0,0).
 \end{split}
\end{equation}

The weighted projective space  $\P^4_{9,6,1,1,1}$ has a curve of $\Z_3$ singularities, located at the points $[z_1 ,z_2 ,0,0,0]$, where $[z_1 , \ldots , z_5 ]$ are homogeneous coordinates  on $\P^4_{9,6,1,1,1}$. Locally these are of the form $\C^3/\Z_3$. We blow up this curve by introducing the ray $v_6$, which satisfies the additional relation
\begin{equation}\label{x1}
 v_ 3+v_4+v_5 -3 v_6 = 0.
\end{equation}

Let $d_i$ denote the divisor corresponding to the ray $v_i$. It is clear from (\ref{x2}) that we have the following linear equivalence relations
\begin{equation}\label{e:lr1}
 d_1\sim 3h,\quad d_2\sim 2h,\quad d_3\sim d_4\sim d_5\sim l,\quad d_6\sim e,\quad
e \sim h - 3l.
\end{equation}

The generic degree $18$ hypersurface in  $\P^4_{9,6,1,1,1}$ intersects the curve of $\Z_3$ singularities in one point, and hence is  singular. Blowing up  $\P^4_{9,6,1,1,1}$ resolves the singularity of the hypersurface as well.
This is achieved torically by considering a generic anti-canonical hypersurface $X$ in $Z$, the  blow-up of $\P^4_{9,6,1,1,1}$ discussed before. It follows from (\ref{x2}) that in the notation of (\ref{e:lr1}) $-K_Z = 6h$, thus $X$ is a generic element in the linear system $ |6h|$.
The exceptional divisor $e$ of $Z$ intersects $X$ in a $\P^2$, which we call $E$, and $X$ is elliptically fibered over $E$. Let $H$ resp. $L$ denote the restriction of the divisors $h$ resp. $l$ of $Z$ to $X$.
As a consequence of (\ref{e:lr1}) and the fact that  $-K_Z = 6h$, on $X$ we have the linear equivalence relations
\begin{equation}\label{e:lr}
 E \sim H - 3L,\quad X\sim 6H.
\end{equation}

This model has four phases (depicted in Fig.~\ref{f:modP}):\footnote{See Section~5 of  \cite{navigation} for more details.}
\begin{enumerate}
 \item the smooth \CY\ phase.
\item an orbifold phase, whose limit point has the orbifold singularity $\C^3/\Z_3$, but the \CY\ has infinite volume.
\item a $\P^2$ phase, where the elliptic fibration $X$ collapses onto its base $\P^2$. In the limit, this elliptic fiber has zero area and the $\P^2$ has infinite volume.
\item a Landau-Ginzburg (LG) phase with the Gepner point as the limit point.
\end{enumerate}

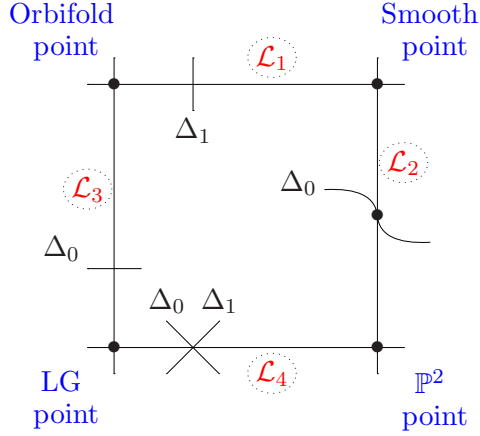
\begin{figure}[h]
\begin{equation}\nonumber
\begin{xy} <3.5cm,0cm>:
{\ar@{-} (-.1,0);(1.1,0) }	,{\ar@{-} (1,.1);(1,-1.1) }	,{\ar@{-} (1.1,-1);(-.1,-1)}	,{\ar@{-} (0,.1);(0,-1.1)} 
,(1.2,.2)*\txt{\color{blue} Smooth\\ \ \color{blue} point}
,(.6,.1)*+[o][F.]{\color{red} {\mathcal L}_1}
,(1.2,-1.2)*\txt{\color{blue} $\P^2$ \\ \ \color{blue} point}
,(1.1,-.3)*+[o][F.]{\color{red} {\mathcal L}_2}
,(-.2,-1.2)*\txt{\color{blue} LG\\ \ \color{blue} point}
,(.6,-1.1)*+[o][F.]{\color{red} {\mathcal L}_4}
,(-.2,.2)*\txt{\color{blue} Orbifold\\ \ \color{blue} point}
,(-.1,-.4)*+[o][F.]{\color{red} {\mathcal L}_3}
,{\ar@{-} (0.3,0.1);(0.3,-.1)   *+!U{\Delta_1}}
,{\ar@{-} (0.2,-1.1);(0.4,-.9)   *+!D{\Delta_1}}		,{\ar@{-} (0.4,-1.1);(0.2,-.9)   *+!D{\Delta_0}}
,{\ar@{-} (.1,-.7);(-.2,-.7)   *+!D{\Delta_0}}		,(0.7,-.3)   *+!U{\Delta_0}
,{(  .8,-.4);(1,-.5) **\crv{(.99,-.39)}}
,{(1.2,-.6);(1,-.5) **\crv{(1.01,-.61)}}
,(0,0)*{\bullet}, (0,-1)*{\bullet} 	,(1,0)*{\bullet} ,(1,-1)*{\bullet} 
,(1,-.5)*{\bullet}	
\end{xy}
\end{equation}
  \caption{The moduli space of the $\P^4_{9,6,1,1,1}[18]$ model.}
  \label{f:modP}
\end{figure}

The discriminant locus of singular CFT's is reducible, with irreducible components $\Delta_0$ and $\Delta_1$. These intersect the four lines of interest in the way depicted in Fig.~\ref{f:modP}: the intersections at ${\mathcal L}_1$ and ${\mathcal L}_3$ are simple transverse; $\Delta_0$ and $\Delta_1 $ meet  ${\mathcal L}_4$ at the same point; while ${\mathcal L}_2$ and $\Delta_0$ meet at third order.

We have three interesting monodromy identities to consider: at the orbifold point, at the $\P^2$ point, and at the LG point. We start with the orbifold point.

\subsubsection{Monodromy around the orbifold point}

It was argued in  \cite{navigation} that monodromy inside ${\mathcal L}_1$ around the orbifold point is given by
\begin{equation}\label{n:1}
 \ms{M}_{1}=\T{i_*\O_E} \comp \L{L},
\end{equation}
where $i$ is the embedding of the exceptional divisor $i\colon E\hookrightarrow X$. Guided by the $\Z_3$ quantum symmetry of the $\C^3/\Z_3$ orbifold, which is created by blowing down $E$, \cite{navigation} shows that
\begin{equation*}
 \ch\left( (\ms{M}_{1})^3(\O_X)\right)=e^H.
\end{equation*}
Now we will show that indeed
\begin{prop}\label{n:8}
 $$(\ms{M}_{1})^3\iso \L{H}.$$
\end{prop}

\begin{proof}
We start by rewriting $(\ms{M}_{1})^3$ in a more convenient form using the definition (\ref{n:1}) and Lemma~\ref{STcomp}:
\begin{equation}\label{n:2}
\begin{split}
 (\ms{M}_{1})^3=&\,(\T{i_*\O_E} \comp \L{L}) \comp (\T{i_*\O_E} \comp \L{L}) \comp (\T{i_*\O_E} \comp \L{L})\\
\iso &\, \T{i_*\O_E} \comp \T{\L{L}(i_*\O_E)} \comp \T{\L{L}^2(i_*\O_E)} \comp \L{L}^3.
\end{split}
\end{equation}
But $\L{L}(i_*\O_E)=i_*(\O_E(i^*L))$, and it is easy to see,  particularly in the toric presentation, that the divisor $L$ restricts to the hyperplane divisor on $E\iso \P^2$, i.e., $\L{L}(i_*\O_E)=i_*\O_E(1)$; and similarly $\L{L}^2(i_*\O_E)=i_*\O_E(2)$. But $\O,\O(1),\O(2)$ is a full exceptional collection on $\P^2$. Using Prop.~\ref{pushout}, eq.~(\ref{n:2}) becomes
\begin{equation*}
 (\ms{M}_{1})^3
\iso \L{\O_X(E)} \comp  \L{L}^3=\L{E+3L}=\L{H}.
\end{equation*}
In the last equality we used the linear equivalence (\ref{e:lr}).
\end{proof}

\subsubsection{Monodromy around the \texorpdfstring{$\P^2$}{P2} point}

It was shown in  \cite{navigation} that monodromy inside ${\mathcal L}_2$ around the $\P^2$ point is given by
\begin{equation}\label{n:4}
 \ms{M}_2=\L{H}\comp \L{L}^{-2}\comp \T{\O_X} \comp \L{L}\comp \T{\O_X} \comp \L{L}\comp \T{\O_X} .
\end{equation}
Passing from the \CY\ point to the $\P^2$ point represents collapsing a large radius elliptic fiber to an LG orbifold theory, which is a $\Z_6$-orbifold, and thus has a $\Z_6$ quantum symmetry. This motivated \cite{navigation} to prove that
\begin{equation*}
 \ch\left( (\ms{M}_2)^6(\O_X)\right)=1.
\end{equation*}
Lifting this to the derived category, we have the following:

\begin{prop}\label{p:excol}
$$\left( \ms{M}_2\right)^6 \iso (-)[2].$$
\end{prop}

\begin{proof}
We first rewrite $\ms{M}_2$ using the definition (\ref{n:4}) and Lemma~\ref{STcomp}:
\begin{equation*}
\ms{M}_2=\L{H} \comp \T{\O_X(-2L)} \comp \T{\O_X(-L)} \comp \T{\O_X} .
\end{equation*}
Then once again using Lemma~\ref{STcomp}
\begin{equation}\label{n:5}
\begin{split}
(\ms{M}_2)^6=\L{H} \comp &\left(\T{\O_X(-2L)} \comp \T{\O_X(-L)} \comp \T{\O_X}\right)
\comp \left(\T{\O_X(H-2L)} \comp \T{\O_X(H-L)} \comp \T{\O_X(H)}\right) \\
&\comp \cdots\comp \left(
\T{\O_X(5H-2L)} \comp \T{\O_X(5H-L)} \comp \T{\O_X(5H)}\right)\comp \L{H}^5 .
\end{split}
\end{equation}

Recall that $Z$ denotes the resolved weighted projective space  $\P^4_{9,6,1,1,1}$, i.e., the ambient space where $X$ is embedded as a smooth hypersurface, as discussed at the beginning of Sec.~\ref{s:P4}, and also recall the toric divisors  $l$ and $h$ on $Z$ from (\ref{e:lr1}). We have the following lemma
\begin{lemma}\label{exColl}
 $$
\O_Z(-2l),\O_Z(-l),{\O_Z},\O_Z(h-2l),\O_Z(h-l),\O_Z(h) , \cdots ,
\O_Z(5h-2l),\O_Z(5h-l),\O_Z(5h)
$$
is a full exceptional collection on $Z$.
\end{lemma}

\begin{proof}[Proof of the lemma.]
First note that both  $\P^4_{9,6,1,1,1}$ and its blow-up $Z$ are singular toric varieties, and hence we need to work with them as smooth stacks. For convenience we tensor each element of the collection by $\O_Z(2l)$, which of course is irrelevant for exceptionality. The task now is to prove that
\begin{equation}\label{e:excoll}
 \O_Z,\O_Z(l),{\O_Z(2l)},\O_Z(h),\O_Z(h+l),\O_Z(h+2l) , \cdots ,
\O_Z(5h),\O_Z(5h+l),\O_Z(5h+2l)
\end{equation}
is exceptional.

The obvious route is to compute the $\Ext$ groups by standard toric techniques available for line bundles on toric varieties, a straightforward but tedious work. Instead, we will use a recent result by Kawamata \cite{Kawamata} which we paraphrase for the convenience of the reader. Section~5 of \cite{Kawamata} considers a toric divisorial contraction $\phi\colon \c X \to \c Y$, with exceptional divisor $\c E$, where $ \c X $ and $ \c Y$ are toric stacks.
$\phi\colon \c X \to \c Y$ is also known as the blow-up map.

Let $\c E_i$ be  the prime divisors on $\c Y$ corresponding to the rays $\{v_i\}_{i=1}^n$ of the toric fan, which in turn define the toric stack $\c Y$. Since $\c X$ is a blow-up, the rays $\{v_i\}_{i=1}^n$ define prime divisors on $\c X$ as well, which we call $\c D_i$. Of course, $\c X$ has one more prime divisor, the exceptional divisor  $\c D_{n+1}=\c E$, corresponding to the additional ray $v_{n+1} $ of the blow-up. The contraction morphism is described by an equation
\begin{equation}\label{relation}
a_1v_1 + \dots + a_{n+1}v_{n+1} = 0,
\end{equation}
for integers $a_i$.

In general there is no morphism of stacks $\mathcal{X}\to \mathcal{Y}$, but there is still a fully faithful functor $\Phi: \D(\mathcal{Y}) \to \D(\mathcal{X})$. Kawamata proves the following isomorphism:
\begin{equation}\label{x4}
 \Phi\left(\mathcal{O}_{\mathcal{Y}}(\sum_{i=1}^n k_i\mathcal{E}_i)\right)
\,\cong \,\mathcal{O}_{\mathcal{X}}\left(\sum_{i=1}^{n} k_i\mathcal{D}_i + N_{\underbar{k}}\c E\right), \quad\mbox{where}\,
 N_{\underbar{k}} = \lfloor \frac{-1}{a_{n+1}} \sum_{i=1}^n  {a_ik_i} \rfloor.
\end{equation}
$\lfloor x\rfloor$ is the integer part of the rational number $x$ (the floor function).

Returning to our problem, it is known that
\begin{equation}\label{x3}
 \O, \O(1),\ldots \O(17)
\end{equation}
is a full exceptional collection on the toric stack $\P^4_{9,6,1,1,1}$. We will now show that (\ref{e:excoll}) is the image of (\ref{x3}) under $\Phi$. The exceptionality of the sequence (\ref{e:excoll}) then follows immediately, since  $\Phi $ is a fully faithful functor. For this we use (\ref{x4}) with the role of (\ref{relation}) being played by  (\ref{x1}):
\begin{equation}\label{x5}
 \Phi(\mathcal{O}_{\P^4_{9,6,1,1,1}}(i))
\,\cong \,\mathcal{O}_{Z}(i l + N_i e), \quad\mbox{where}\,
 N_i = \lfloor \frac{1}{3} i \rfloor.
\end{equation}
Running through the index set $0,1,\ldots, 17$, and 
using the linear equivalence $e \sim h - 3l$ from (\ref{e:lr1}) completes the proof of exceptionality.

The fact that (\ref{e:excoll}) is full follows from  Theorem~5.2(1) of \cite{Kawamata}.
\end{proof}

Returning to the proof of Prop.~\ref{p:excol}, in particular Eq.~(\ref{n:5}), first recall that $L$ and $H$ are the restrictions of the toric divisors  $l$ and $h$ on $Z$. Therefore,
$$
\O_X(\alpha H + \beta L) = j^* \O_Z(\alpha h + \beta l),\quad \mbox{for all $\alpha,\beta\in \Z$},
$$
where $j\colon X\to Z$ is the embedding.
Using Lemma~\ref{pullback} and Lemma~\ref{exColl},  Eq.~(\ref{n:5}) becomes
\begin{equation*}
 (\ms{M}_2)^6\iso \L{H} \comp  \L{\O_X(-X)[2]} \comp \L{H}^5 =\L{6H-X}[2]=(-)[2].
\end{equation*}
In the last equality we used again the linear equivalence (\ref{e:lr}).
\end{proof}

\subsubsection{Monodromy around the LG point}

Monodromy inside ${\mathcal L}_3$ around the LG point  was shown    to be \cite{navigation}
\begin{equation}\label{n:6}
 \ms{M}_3=\T{\O_X} \comp \ms{M}_1.
\end{equation}
The fact that the LG point is a $\Z_{18}$  orbifold motivated \cite{navigation} to prove that
\begin{equation*}
 \ch\left( (\ms{M}_3)^{18}(\O_X)\right)=1.
\end{equation*}
But in fact more is true:
\begin{prop}\label{m1}
 $$\left( \ms{M}_3\right)^{18}\iso (-)[2].$$
\end{prop}

\begin{proof}
Using the definition (\ref{n:6}) and Lemma~\ref{STcomp} we have that
\begin{equation*}
\begin{split}
(\ms{M}_3)^3=&\,(\T{\O_X} \comp \ms{M}_1) \comp (\T{\O_X} \comp \ms{M}_1) \comp (\T{\O_X} \comp \ms{M}_1) \\
\iso & \,\T{\O_X} \comp \T{\ms{M}_1(\O_X)} \comp \T{\ms{M}_1^2(\O_X)}  \comp \ms{M}_1^3.
\end{split}
\end{equation*}
Using Prop.~\ref{n:8}, we have that
\begin{equation}\label{n:7}
(\ms{M}_3)^3\iso  \T{\O_X} \comp \T{\ms{M}_1(\O_X)} \comp \T{\ms{M}_1^2(\O_X)}  \comp \L{H}.
\end{equation}

To proceed, we need to compute $\ms{M}_1(\O_X)$ and $ \ms{M}_1^2(\O_X)$. Let's start with  $\ms{M}_1(\O_X)$. From it's definition in (\ref{n:1}), $\ms{M}_1(\O_X)=\T{i_*\O_E} (\O_X(L))$. From the definition (\ref{e:refl}) it's clear that we need to compute $\Hom_{\D(X)}(i_*\O_E,\O_X(L))$. For later convenience we compute  $\Hom_{\D(X)}(i_*\O_E,\O_X(k L))$ for all $k\in \Z$.

\begin{lemma}
$ \Hom_{\D(X)}^a(i_*\O_E,\O_X(k L)) = \H^a(\P^2,\O_{\P^2}(k-3)) $ for all $a \in \Z$, and  $k \in \Z$.
\end{lemma}

\begin{proof}[Proof of the lemma.]
First observe that $i_*\O_E=\R i_*\O_E$, and use the fact that $i^!$ is the right adjoint functor of $\R i_* $:
\begin{equation}\label{v1}
 \Hom_{\D(X)}(i_*\O_E,\O_X(L)) = \Hom_{\D(E)}(\O_E,i^! \O_X(L)).
\end{equation}
On the other hand, for $i\colon E\hookrightarrow X$ an embedding of a divisor, and $X$ a \CY ,
\begin{equation}\label{v2}
   i^! \O_X(k L) = \LL i^* \O_X(k L) \Ltensor \omega_E \iso \O_E(k-3).
\end{equation}
(see, e.g., \cite{Horj:EZ} or the Appendix of \cite{en:Ema} for some properties of $i^!$).

Since $E\iso \P^2$, (\ref{v1}) and (\ref{v2}) imply that
\begin{equation*}
 \Hom_{\D(X)}(i_*\O_E,\O_X(k L)) = \H^*(\P^2,\O_{\P^2}(k-3)).
\end{equation*}
\end{proof}

Using the lemma, and the fact that $\H^i(\P^2,\O_{\P^2}(-2))=\H^i(\P^2,\O_{\P^2}(-1))=0$ for all $i\in \Z$, we immediately see that
$$
\ms{M}_1(\O_X)=\O_X(L),\quad
 \ms{M}_1^2(\O_X)=\O_X(2L).
$$
Thus (\ref{n:7}) becomes
\begin{equation*}
(\ms{M}_3)^3\iso  \left(\T{\O_X} \comp   \T{\O_X(L)} \comp   \T{\O_X(2L)} \right)  \comp \L{H}.
\end{equation*}
Repeatedly using Lemma~\ref{STcomp} (by moving $\L{H}$ to the right) gives
\begin{equation*}
\begin{split}
(\ms{M}_3)^{18}\iso &\, \left(\T{\O_X} \comp   \T{\O_X(L)} \comp   \T{\O_X(2L)} \right)  \comp
 \left(\T{\O_X(H)} \comp   \T{\O_X(H+L)} \comp   \T{\O_X(H+2L)} \right)\\
&\comp \cdots \comp
 \left(\T{\O_X(5H)} \comp   \T{\O_X(5H+L)} \comp   \T{\O_X(5H+2L)} \right)  \comp
\L{H}^6. 
\end{split}
\end{equation*}
But this is the exceptional collection (\ref{e:excoll}) in the proof of Lemma~\ref{exColl}, and therefore Lemma~\ref{pullback} gives that
\begin{equation*}
(\ms{M}_3)^{18}\iso \L{\O_X(-X)[2]} \comp \L{H}^6 =\L{6H-X}[2]=(-)[2].
\end{equation*}
\end{proof}

Let us mention that the direct approach used in proving Prop.~\ref{m1} does not work for proving (\ref{c22}) or Prop.~\ref{c32}. The reason is that we have sheaves supported on different divisors in those cases, and the analogs of the steps in the proof of Prop.~\ref{m1} do not lead to an exceptional collection, and we are led to use a cleverer approach.

\section{Discussion}

In the light of our results, the reader will naturally ask the question of how generic are these type of results? On the physics side, given a \CY\ compactification, we expect discrete symmetries to arise at various points in moduli space, and hence the moduli space locally is an orbifold. Therefore the fundamental group of the moduli space is  non-trivial. Choosing a presentation, the generators of this group will satisfy certain relations. D-branes, through their monodromy, translate these relations into relations between the associated autoequivalences. As a result, we expect interesting identities between  autoequivalences for a general \CY\ compactification.

The next question is how often do these identities follow from the technique of exceptional collections developed here. Obviously, we cannot expect to have an exceptional collection for every instance. In this paper we looked at three examples with two dimensional moduli space, two local and one compact, and proved a total of nine identities. Every one of them followed from the existence of an exceptional collection (either on a divisor, or in the ambient space). The author also studied the compact models
$\P^4_{1,1,2,2,2}$ and  $\P^4_{1,1,2,8,12}$, with two resp. three dimensional Kahler moduli space,
and obtained similar results.

On the other hand, most known \CY\ varieties are subvarieties in toric varieties (hypersurfaces and complete intersections). Kawamata  proves (Theorem 1.1. of \cite{Kawamata:DC}) that if $X$ is a projective toric variety with at most quotient
singularities, then the associated smooth Deligne-Mumford stack $\c X$ has a complete exceptional collection consisting
of sheaves. This suggests that for large class of  \CY\ varieties at least part of the monodromy identities should indeed follow from the exceptional collection techniques. It would be interesting to study examples where the ambient space does not have an exceptional collection.

Unfortunately general statements are beyond reach at this point, even in the toric case. The first obstacle is that in order to write down  the identities we need a detailed understanding of the singularities of the moduli space, the discriminant loci, and its intersections with the large radius divisors. All of these are hard to get, but  computing the discriminant in the general case seems impossible. There are two ways to approach the problem:
\begin{enumerate}
 \item write down the equations enforcing that the mirror is singular, and use elimination theory.
\item use  Horn parametrization, and then elimination theory (see, e.g., \cite{en:fracC2}).
\end{enumerate}
Using Groebner basis for the  elimination part, both approaches give the same answer, but as the dimension of the moduli space increases, today's computers are unable to solve the elimination problem. Even if we knew the discriminant, its intersections with the large radius divisors is very diverse  (we already saw evidence in our examples), hence writing down general statements seems impossible. Therefore, the best one can do is to prove the identities case by case, as we have done in this paper. On the other hand, this is not an unsatisfactory state of affairs, since this is the best one can do for the proof of mirror symmetry as well: prove it case by case.

One family of examples where we can make general statements is the case of hypersurfaces in weighted projective spaces. Physics-wise the relevant case is the 3-fold, and it is meaningless to work on 6-folds and higher, but mathematically the statement holds for all weights and all dimensions. The weighted projective space has only quotient singularities, and also has a strong and full exceptional collection when viewed as a stack. As discussed in the introduction, the relevant identity was proven in full generality in \cite{en:Alberto}.

There are two future directions that our technique seems suitable to tackle. The first is to extend our results to complete intersections in toric varieties. The second is the connection with Horja's EZ-transformations \cite{Horj:EZ}. The conjectured autoequivalences were proved only in the case when the \CY\ is a fibration over the projective space of dimension $d$, $\P^d$ \cite{en:Horja}. The proof used the so-called ``Verdier 9-diagram'', and had no mention of exceptional collections. Canonaco  \cite{AC} pointed out  that the same result also follows from the exceptional collection approach. This opens up the possibility of proving other cases of Horja's EZ-conjecture.

\section*{Acknowledgments}

It is a pleasure to thank Paul Aspinwall, Tom Bridgeland, Mike Douglas, Alastair King and Eric Sharpe for useful conversations. I am especially indebted to Alberto Canonaco for collaboration on a related project. This work was partially supported by NSF grant PHY-0755614.

\vspace{1cm}


\begin{thebibliography}{10}

\bibitem{Aspinwall:2004bs}
P.~S. Aspinwall and S.~Katz,
\newblock {\em Computation of superpotentials for D-Branes},
\newblock Commun. Math. Phys. {\bf 264} (2006) 227--253, hep-th/0412209.

\bibitem{en:Ema}
D.-E. Diaconescu, A.~Garcia-Raboso, R.~L. Karp, and K.~Sinha,
\newblock {\em D-Brane Superpotentials in Calabi-Yau Orientifolds},
\newblock Adv. Theor. Math. Phys. {\bf 11} (2007) 471--516, hep-th/0606180.

\bibitem{Douglas:2000gi}
M.~R. Douglas,
\newblock {\em D-branes, categories and N = 1 supersymmetry},
\newblock J. Math. Phys. {\bf 42} (2001) 2818--2843, hep-th/0011017.

\bibitem{Sharpe:1999qz}
E.~R. Sharpe,
\newblock {\em {D-branes, derived categories, and Grothendieck groups}},
\newblock Nucl. Phys. {\bf B561} (1999) 433--450, hep-th/9902116.

\bibitem{DBook}
P.~S. Aspinwall et~al.,
\newblock {\em \em Dirichlet branes and mirror symmetry},
\newblock Clay mathematics monographs, To appear.

\bibitem{Paul:TASI2003}
P.~S. Aspinwall,
\newblock {\em D-branes on Calabi-Yau manifolds},
\newblock in ``Recent Trends in String Theory'', pages 1--152, World
  Scientific, 2004,
\newblock hep-th/0403166.

\bibitem{Sharpe:2006vd}
E.~Sharpe,
\newblock {\em {Derived categories and stacks in physics}} (2006),
\newblock hep-th/0608056.

\bibitem{Cox:Katz}
D.~A. Cox and S.~Katz,
\newblock {\em Mirror symmetry and algebraic geometry}, Mathematical Surveys
  and Monographs~{\bf 68},
\newblock AMS, Providence, RI, 1999.

\bibitem{en:Horja}
P.~S. Aspinwall, R.~L. Karp, and R.~P. Horja,
\newblock {\em Massless D-branes on Calabi-Yau threefolds and monodromy},
\newblock Commun. Math. Phys. {\bf 259} (2005) 45--69, hep-th/0209161.

\bibitem{Horj:EZ}
R.~P. Horja,
\newblock {\em Derived category automorphisms from mirror symmetry},
\newblock Duke Math. J. {\bf 127} (No. 1)  (2005) 1--34,
  math\-.\-AG\-/\-0103231.

\bibitem{en:Alberto}
A.~Canonaco and R.~L. Karp,
\newblock {\em Derived autoequivalences and a weighted Beilinson resolution},
\newblock J. Geom. Phys. {\bf 58} (2008) 743–760, math.AG/0610848.

\bibitem{AC}
A.~Canonaco,
\newblock {\em Exceptional sequences and derived autoequivalences},
\newblock arXiv:0801.0173.

\bibitem{Zaslow:1994nk}
E.~Zaslow,
\newblock {\em {Solitons and helices: The Search for a math physics bridge}},
\newblock Commun. Math. Phys. {\bf 175} (1996) 337--376, hep-th/9408133.

\bibitem{Hori:2000ck}
K.~Hori, A.~Iqbal, and C.~Vafa,
\newblock {\em {D-branes and mirror symmetry}} (2000),
\newblock hep-th/0005247.

\bibitem{Govindarajan:2000vi}
S.~Govindarajan and T.~Jayaraman,
\newblock {\em {D-branes, exceptional sheaves and quivers on Calabi-Yau
  manifolds: From Mukai to McKay}},
\newblock Nucl. Phys. {\bf B600} (2001) 457--486, hep-th/0010196.

\bibitem{Tomasiello:2000ym}
A.~Tomasiello,
\newblock {\em {D-branes on Calabi-Yau manifolds and helices}},
\newblock JHEP {\bf 02} (2001) 008, hep-th/0010217.

\bibitem{Mayr:2000as}
P.~Mayr,
\newblock {\em {Phases of supersymmetric D-branes on Kaehler manifolds and the
  McKay correspondence}},
\newblock JHEP {\bf 01} (2001) 018, hep-th/0010223.

\bibitem{Cachazo:2001sg}
F.~Cachazo et~al.,
\newblock {\em A geometric unification of dualities},
\newblock Nucl. Phys. {\bf B628} (2002) 3--78, hep-th/0110028.

\bibitem{Wijnholt:2002qz}
M.~Wijnholt,
\newblock {\em {Large volume perspective on branes at singularities}},
\newblock Adv. Theor. Math. Phys. {\bf 7} (2004) 1117--1153, hep-th/0212021.

\bibitem{Douglas:Moore}
M.~R. Douglas and G.~W. Moore,
\newblock {\em D-branes, Quivers, and ALE Instantons},
\newblock hep-th/9603167.

\bibitem{Herzog:2005sy}
C.~P. Herzog and R.~L. Karp,
\newblock {\em Exceptional collections and D-branes probing toric
  singularities},
\newblock JHEP {\bf 02} (2006) 061, hep-th/0507175.

\bibitem{en:Chris2}
C.~P. Herzog and R.~L. Karp,
\newblock {\em On the geometry of quiver gauge theories (Stacking exceptional
  collections)},
\newblock hep-th/0605177.

\bibitem{en:fracC3}
R.~L. Karp,
\newblock {\em On the ${\mathbb C}^n/{\mathbb Z}_m$ fractional branes},
\newblock hep-th/0602165.

\bibitem{Orlov:96}
D.~O. Orlov,
\newblock {\em Equivalences of derived categories and {$K3$} surfaces},
\newblock J. Math. Sci. (New York) {\bf 84} (No. 5)  (1997) 1361--1381,
  math.AG/9606006.

\bibitem{Kawamata:DC}
Y.~Kawamata,
\newblock {\em Equivalences of derived categories of sheaves on smooth stacks},
\newblock Amer. J. Math. {\bf 126} (No. 5)  (2004) 1057--1083, math.AG/0210439.

\bibitem{ST:braid}
P.~Seidel and R.~Thomas,
\newblock {\em Braid group actions on derived categories of coherent sheaves},
\newblock Duke Math. J. {\bf 108} (No. 1)  (2001) 37--108, math.AG/0001043.

\bibitem{en:fracC2}
R.~L. Karp,
\newblock {\em ${\mathbb C}^2/{\mathbb Z}_n$ fractional branes and monodromy},
\newblock Commun. Math. Phys. {\bf 270} (2007) 163--196, hep-th/0510047.

\bibitem{Huybrechts}
D.~Huybrechts,
\newblock {\em Fourier-{M}ukai transforms in algebraic geometry},
\newblock Oxford Mathematical Monographs, The Clarendon Press Oxford University
  Press, Oxford, 2006.

\bibitem{navigation}
P.~S. Aspinwall,
\newblock {\em Some navigation rules for D-brane monodromy},
\newblock J. Math. Phys. {\bf 42} (2001) 5534--5552, hep-th/0102198.

\bibitem{Kawamata}
Y.~Kawamata,
\newblock {\em Derived categories of toric varieties},
\newblock Michigan Math. J. {\bf 54} (No. 3)  (2006) 517--535,
  arXiv:math.AG/0503102.

\end{thebibliography}
\end{document}